\documentclass[twoside,11pt]{article}
\usepackage{pgm}
\usepackage[utf8]{inputenc}
\inputencoding{utf8}

\usepackage{hyperref,color,soul,booktabs}
\usepackage{threeparttable}
\setulcolor{blue}
\newcommand{\BibTeX}{\textsc{B\kern-0.1emi\kern-0.017emb}\kern-0.15em\TeX}

\usepackage{siunitx}
\usepackage{booktabs}
\usepackage[linesnumbered,algoruled,boxed,lined,ruled]{algorithm2e}

\ShortHeadings{On the Properties of MVR Chain Graphs}{Javidian and Valtorta}

\begin{document}

\title{On the Properties of MVR Chain Graphs}
\author{\Name{Mohammad Ali Javidian} \Email{javidian@email.sc.edu}\and
	\Name{Marco Valtorta} \Email{mgv@cse.sc.edu}\\
	\addr Department of Computer Science \& Engineering, University of South Carolina, Columbia, SC, 29201, USA.}

\maketitle

\begin{abstract}
Depending on the interpretation of the type of edges, a chain graph can represent different relations between variables and thereby independence models. Three interpretations, known by the acronyms LWF, MVR, and AMP, are prevalent. We review Markov properties for MVR chain graphs and propose an alternative local Markov property for them. Except for pairwise Markov properties, we show that for MVR chain graphs all Markov properties in the literature are equivalent for semi-graphoids. We derive a new factorization formula for MVR chain graphs which is more explicit than and different from the proposed factorizations for MVR chain graphs in the literature. Finally, we provide a summary table comparing different features of LWF, AMP, and MVR chain graphs.
\end{abstract}

\section{Introduction}
A probabilistic graphical model is a probabilistic model for which a graph represents the conditional dependence structure between random variables. There are several classes of graphical
models; Bayesian networks (BN), Markov networks, chain graphs, and ancestral graphs are commonly used \citep{l, rs}. Chain graphs, which admit both directed and undirected edges,  are a type of graphs in which there are no partially directed cycles.  Chain graphs were introduced by Lauritzen, Wermuth and Frydenberg \citep{f, lw} as a generalization of graphs based on undirected graphs and directed acyclic graphs (DAGs). Later Andersson, Madigan 
and Perlman introduced an alternative Markov property for chain graphs \citep{amp}.
In 1993 \citep{cw1}, Cox and Wermuth introduced multivariate regression chain graphs (MVR CGs).

Acyclic directed mixed graphs (ADMGs), also known as semi-Markov(ian) \citep{pj} 
models contain directed ($\rightarrow$) and bidirected
($\leftrightarrow$) edges subject to the restriction that there are no directed cycles \citep{r2, er}.
An ADMG that has no partially directed cycle is called a \textit {multivariate regression chain graph}. In this paper we focus on the class of multivariate regression chain graphs and we discuss their Markov properties. 

It is worthwhile to mention that unlike in the other CG interpretations, bidirected edges in an MVR CG have
a strong intuitive meaning. It can be seen that a bidirected edge represents one or more hidden
common causes between the variables connected by it. In other words, in an MVR CG any bidirected
edge $X\leftrightarrow Y$ can be replaced by $X\gets H\to Y$ to obtain a Bayesian network representing
the same independence model over the original variables, i.e. excluding the
new variables H. These variables are called hidden, or latent, and have been
marginalized away in the CG model \citep{s}. This causal interpretation of bidirected edges in MVR CGs along with the discussion preceding Theorem \ref{th1} provides strong motivation for the importance of MVR CGs. 

In the first decade of the 21st century, several Markov property (global, pairwise, block recursive, and so on) were introduced by authors and researchers \citep{rs, wc, ml1, ml2, d}. Lauritzen, Wermuth, and Sadeghi \citep{sl, sw} proved that the global and (four) pairwise Markov properties of an MVR chain graph are equivalent for any independence model that is a compositional graphoid. The major contributions of this paper may be summarized as follows:

	\noindent $\bullet$ An alternative local Markov property for MVR chain graphs, which is equivalent to other Markov properties in the literature for compositional semi-graphoids.
	
	\noindent $\bullet$ A comparison of different proposed Markov properties for MVR chain graphs in the literature and conditions under which they are equivalent. 
	
	\noindent $\bullet$ An alternative explicit factorization criterion for MVR chain graphs based on the proposed factorization criterion for acyclic directed mixed graphs in \citep{er}.
    
\section{Definitions and Concepts}
\begin{definition}
	A vertex $\alpha$ is said to be an \emph{ancestor} of a vertex $\beta$ if either there is a directed path $\alpha \to \dots \to \beta$ from $\alpha$ to $\beta$, or $\alpha=\beta$. A vertex $\alpha$ is said to be \emph{anterior} to a vertex $\beta$ if there is a path $\mu$ from $\alpha$ to $\beta$ on which every edge is either of the form $\gamma-\delta$, or $\gamma \to\delta$ with $\delta$ between $\gamma$ and $\beta$, or $\alpha=\beta$; that is, there are no edges $\gamma\leftrightarrow\delta$ and there are no edges $\gamma\gets\delta$ pointing toward $\alpha$. Such a path is said to be an anterior path from $\alpha$ to $\beta$.   We apply these definitions disjunctively to sets: $an(X) = \{\alpha | \alpha \textrm{ is an ancestor of } \beta \textrm{ for some } \beta \in X\}$, and 	
	$ant(X) = \{\alpha | \alpha \textrm{ is an anterior of } \beta \textrm{ for some } \beta \in X\}$. If necessary we specify the graph by a subscript, as in $ant_G(X)$. 
	The usage of the terms “ancestor” and “anterior” differs from Lauritzen \citep{l}, but follows Frydenberg \citep{f}.
\end{definition}
\begin{definition}\label{ances}
	A mixed graph is a graph containing three types of edges, undirected ($-$), directed ($\to$) and bidirected ($\leftrightarrow$). An ancestral graph G is a mixed graph in which the following conditions hold for all vertices $\alpha$ in G:
	
		\noindent (i) if $\alpha$ and $\beta$ are joined by an edge with an arrowhead at $\alpha$, then $\alpha$ is not anterior to $\beta$.
		
		\noindent (ii) there are no arrowheads present at a vertex which is an endpoint of an undirected edge.
\end{definition}
\begin{definition}\label{RS}
	A nonendpoint vertex $\zeta$ on a path is a \emph{collider} on the path if the edges preceding and succeeding $\zeta$ on the path have an arrowhead at $\zeta$, that is, $\to \zeta \gets, or \leftrightarrow \zeta \leftrightarrow, or\leftrightarrow \zeta \gets, or\to \zeta \leftrightarrow$. A nonendpoint vertex $\zeta$ on a path which is not a collider is a noncollider on the path. A path between vertices $\alpha$ and $\beta$ in an ancestral graph G is said to be m-connecting given a set Z (possibly empty), with $\alpha, \beta \notin Z$, if: 
	
		\noindent (i) every noncollider on the path is not in Z, and 
		
		\noindent (ii) every collider on the path is in $ant_G(Z)$.
	
	If there is no path m-connecting $\alpha$ and $\beta$ given Z, then $\alpha$ and $\beta$ are said to be \emph{$m$-separated} given Z. Sets X and Y are m-separated given Z, if for every pair $\alpha, \beta$, with $\alpha\in X$ and $\beta \in Y$, $\alpha$ and $\beta$ are $m$-separated given $Z$ (X, Y, and Z are disjoint sets; X, Y are nonempty). This criterion is referred to as a \emph{global Markov property}. We denote the independence model resulting from applying the m-separation criterion to G, by $\Im_m$(G). This is an extension of Pearl's $d$-separation criterion to mixed graphs in that in a DAG D, a path is $d$-connecting if and only if it is m-connecting.
\end{definition}
\begin{definition}
	Let $G_A$ denote the induced subgraph of $G$ on the vertex set $A$, formed by removing from $G$ all vertices that are not in $A$, and all edges that do not have both endpoints in $A$. Two vertices $x$ and $y$ in an MVR chain graph $G$ are said to be collider connected if there is a path from $x$ to $y$ in $G$ on which every non-endpoint vertex is a collider; such a path is called a collider path. (Note that a single edge trivially forms a collider path, so if $x$ and $y$ are adjacent in an MVR chain graph then they are collider connected.) The \emph{augmented graph} derived from $G$, denoted $(G)^a$, is an undirected graph with the same vertex set as $G$ such that $c\--d \textrm{ in } (G)^a \Leftrightarrow c \textrm{ and } d \textrm{ are collider connected in } G.$
\end{definition}

\begin{definition}\label{RS2}
	Disjoint sets $X, Y\ne \emptyset,$ and $Z$ ($Z$ may be empty) are said to be
	\emph{$m^\ast$-separated} if $X$ and $Y$ are separated by Z in $(G_{ant(X\cup Y\cup Z)})^a$. Otherwise $X$ and $Y$ are said to be $m^\ast$-connected
	given $Z$. The resulting independence model is denoted by $\Im_{m^\ast}(G)$.
\end{definition}
Richardson and Spirtes in \citep[Theorem 3.18.]{rs} show that for an ancestral graph $G$, $\Im_m(G)=\Im_{m^\ast}(G)$. Note that in the case of ADMGs and MVR CGs, anterior sets in definitions \ref{RS}, \ref{RS2} can be replaced by ancestor sets, because in both cases anterior sets and ancestor sets are the same.

\begin{definition}
	An ancestral graph G is said to be maximal if for every pair of vertices $\alpha, \beta$ if $\alpha$ and $\beta$ are not adjacent in G then there is a set Z ($\alpha, \beta\notin Z$), such that $\langle\{\alpha\},\{\beta\} | Z\rangle  \in \Im_m (G)$. Thus a graph is maximal if every missing edge corresponds to at least one independence in the corresponding independence model.
\end{definition}

A simple example of a nonmaximal ancestral graph is shown in Figure \ref{Fig:6}: $\gamma$ and $\delta$ are not adjacent, but are $m$-connected given every subset of $\{\alpha,\beta\}$, hence $\Im_m(G) =\emptyset$.
	\begin{figure}[ht]
	\centering
	\includegraphics[scale=.75]{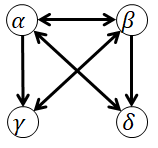}
	\caption{\citep{rs} A nonmaximal ancestral graph.} \label{Fig:6}
\end{figure}

If $G$ is an undirected graph or a directed acyclic graph, then $G$ is a maximal ancestral graph \citep[Proposition 3.19]{rs}.

The absence of partially directed cycles in MVR CGs implies that the vertex set of a chain graph
can be partitioned into so-called chain components such that edges within a chain component are
bidirected whereas the edges between two chain components are directed and point in the same
direction. So, any chain graph yields a directed acyclic graph $D$ of its chain components having $\mathcal{T}$ as a node set and an edge $T_1\to T_2$ whenever there exists in the chain graph $G$ at least one edge $u\rightarrow v$ connecting a node \textit{u} in $T_1$ with a node \textit{v} in $T_2$. In this directed graph, we may define for each $T$ the set $pa_D (T)$ as the union of all the chain components that are parents of $T$ in the directed graph $D$. This concept is distinct from the usual notion of the parents $pa_G(A)$ of a set of nodes $A$ in the chain graph, that is, the set of all the nodes $w$ outside $A$ such that $w\to v$ with $v\in A$ \citep{ml2}.

Given a chain graph G with chain components $(T|T\in \mathcal{T})$, we can always define a strict total order $\prec$ of the chain components that is consistent with the partial order induced by the chain graph, such that if $T \prec T'$ then $T\notin pa_D(T')$ (we draw $T'$ to the right of $T$ as in the example of Figure \ref{Fig:mvr}).

\begin{figure}[ht]
		\centering
		\includegraphics[scale=.5]{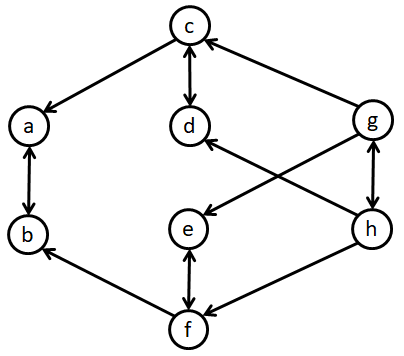}
		\caption{An MVR CG with chain components: $\mathcal{T}=\{T_1=\{a,b\} , T_2=\{c,d\}, T_3=\{e,f\}, T_4=\{g,h\}\}$.} \label{Fig:mvr}
	\end{figure}
    
For each $T$, the set of all components preceding $T$ is known and we may define the cumulative set $pre(T)=\cup_{T\prec T'}T'$  of nodes contained in the predecessors of component $T$, which we sometimes call the past of $T$. The set $pre(T)$ captures the notion of all the potential explanatory variables of the response variables within $T$ \citep{ml2}. In fact, MVR CGs can model the possible presence of residual
associations among the responses using a bidirected
graph, and this is
consistent with an interpretation of bidirected edges in terms of latent variables \citep{rov,e}.

\section{Markov Properties for MVR Chain Graphs}
In this section, first, we show,  formally, that MVR chain graphs are a subclass of the maximal ancestral graphs of Richardson and Spirtes \citep{rs} that include only observed and latent variables. Latent variables cause several complications. First, causal inference based on structural learning algorithms such as the PC algorithm \citep{sgs} may be incorrect. Second, if a distribution is faithful to a DAG, then the distribution obtained by marginalizing out on some of the variables may not be faithful to any DAG on the observed variables i.e., the space of DAGs is not closed under marginalization \citep{cmkr}.
\begin{example}\label{example1}
	Consider that the DAG $G$ in Figure \ref{Fig:ex1}(a) is a perfect map of the distribution of $(X, Y, U, V, H)$, and suppose that $H$ is latent. There is no DAG on $\{X, Y, U, V\}$ that encodes exactly the
	same d-separation relations among $\{X, Y, U, V\}$ as $G$. Hence, there does not exist a perfect map of the
	marginal distribution of $(X, Y, U, V, H)$. 
\end{example}
	\begin{figure}[ht]
	\centering
	\includegraphics[scale=.75]{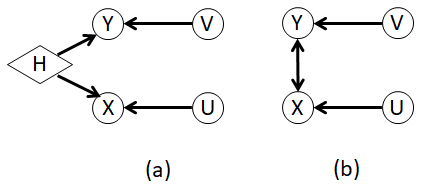}
	\caption{(a) A directed graph including a vertex $H$ for an unobserved variable, (b) the
		independence structure encoded by the MVR CG} \label{Fig:ex1}
\end{figure}
Mixed graphs provide a useful approach to address these problems without explicit modeling
of latent variables (e.g., \citep{rs,pj,ws}). The nodes of these
graphs index the observed variables only. The edges, however, may be of two types, directed and
bidirected. This added flexibility allows one to represent the more complicated dependence structures arising from a DAG with latent variables. A straightforward generalization of d-separation
determines conditional independencies in mixed graph models \citep{dm}. For instance, the MVR chain graph in Figure \ref{Fig:ex1} (b) is a perfect map for the distribution in Example \ref{example1}.
As a result, one possibility for solving the above mentioned problems is exploiting MVR chain graphs that cope with these problems without explicit modeling
of latent variables. This motivates the development of studies on MVR CGs, and \citep{dm} emphasize that methods that account for the effects of latent variables need to be developed further.   
\begin{theorem}\label{th1}
	If $G$ is an MVR chain graph, then $G$ is an ancestral graph.
\end{theorem}
\begin{proof}
	Obviously, every MVR chain graph is a mixed graph without undirected edges. So, it is enough to show that condition (i) in Definition \ref{ances} is satisfied. For this purpose, consider that $\alpha$ and $\beta$ are joined by an edge with an arrowhead at $\alpha$ in MVR chain graph G. Two cases are possible. First, if $\alpha\leftrightarrow\beta$ is an edge in G, by definition of an MVR chain graph, both of them belong to the same chain component. Since all edges on a path between two nodes of a chain component are bidirected, then by definition $\alpha$ cannot be an anterior of $\beta$. Second, if $\alpha\leftarrow\beta$ is an edge in G, by definition of an MVR chain graph, $\alpha$ and $\beta$ belong to two different components ($\beta$ is in a chain component that is to the right side of the chain component that contains $\alpha$). We know that all directed edges in an MVR chain graph are arrows pointing from right to left, so there is no path from $\alpha$ to $\beta$ in G i.e. $\alpha$ cannot be an anterior of $\beta$ in this case. We have shown that $\alpha$ cannot be an anterior of $\beta$ in both cases, and therefore condition (i) in Definition \ref{ances} is satisfied. In other words, every MVR chain graph is an ancestral graph. 
\end{proof}

The following result is often mentioned in the literature \citep{ws, p2, sl, s}, but we know of no published proof.
\begin{corollary}\label{cor1}
	Every MVR chain graph has the same independence model as a DAG under marginalization.
\end{corollary}
\begin{proof}
	From Theorem \ref{th1}, we know that every MVR chain graph is an ancestral graph. The result follows directly from \citep[Theorem 6.3]{rs}.
\end{proof}

\begin{corollary}\label{th2}
	If $G$ is an MVR chain graph, then $G$ is a maximal ancestral graph.
\end{corollary}
\begin{proof}
To characterize maximal ancestral graphs, we need the following notion: A chain $<r, q_1,\cdots,q_p,s>$ is a
primitive inducing chain between $r$ and $s$ if and only if for every $i$, $1\le i\le p$:
\begin{itemize}
\item $q_i$ is a collider on the chain; and
\item $q_i \in an(\{r\} \cup \{s\}).$
\end{itemize}
Based on Corollary 4.4 in \citep{rs}, every nonmaximal ancestral graph contains a primitive
inducing chain between a pair of nonadjacent vertices. So, it is enough to show that an MVR chain graph $G$ does not contain a primitive
inducing chain between any pair of nonadjacent vertices of $G$. For this purpose, consider that $r$ and $s$ are a pair of nonadjacent vertices in MVR chain graph $G$ such that chain $<r, q_1,\cdots,q_p,s>$ is a
primitive inducing chain between $r$ and $s$. So, for every $i$, $1\le i\le p$: $q_i$ is a collider on the chain. Since, for every $i$, $1\le i\le p$: $q_i \in an(\{r\} \cup \{s\})$, there is a partially directed cycle in $G$, which is a contradiction. 
\end{proof}

\subsection{Global and Pairwise Markov Properties}
The following properties have been defined for conditional independences of probability distributions. Let $A, B, C$ and $D$
be disjoint subsets of $V_G$, where $C$ may be the empty set.

\noindent 1. Symmetry: $A\!\perp\!\!\!\perp B \Rightarrow B \!\perp\!\!\!\perp A$;

\noindent 2. Decomposition: $A\!\perp\!\!\!\perp BD | C \Rightarrow (A\!\perp\!\!\!\perp B | C \textrm{ and }  A\!\perp\!\!\!\perp D | C)$;

\noindent 3. Weak union: $A\!\perp\!\!\!\perp BD | C \Rightarrow (A\!\perp\!\!\!\perp B | DC \textrm{ and }  A\!\perp\!\!\!\perp D | BC)$;

\noindent 4. Contraction: $(A\!\perp\!\!\!\perp B | DC \textrm{ and }  A\!\perp\!\!\!\perp D | C) \Leftrightarrow A\!\perp\!\!\!\perp BD | C$;

\noindent 5. Intersection: $(A\!\perp\!\!\!\perp B | DC \textrm{ and }  A\!\perp\!\!\!\perp D | BC) \Rightarrow A\!\perp\!\!\!\perp BD | C$;

\noindent 6. Composition: $(A\!\perp\!\!\!\perp B | C \textrm{ and }  A\!\perp\!\!\!\perp D | C) \Rightarrow A\!\perp\!\!\!\perp BD | C$.
An independence model is a \textit{semi-graphoid} if it satisfies the first four independence properties listed above. Note that every probability distribution $p$ satisfies the \textit{semi-graphoid} properties \citep{studeny1}. If a semi-graphoid further satisfies the intersection property, we say it is a \textit{graphoid} \citep{pp,studeny,studeny1}. A \textit{compositional graphoid} further satisfies the composition property \citep{sw}. If a semi-graphoid further satisfies the composition property, we say it is a \textit{compositional semi-graphoid}.

For a node $i$ in the connected component $T$, its \textit{past}, denoted by $pst(i)$, consists of all nodes in components having a higher order than $T$. To define pairwise Markov properties for MVR CGs, we use the following notation for parents, anteriors
and the past of node pair $i, j$:
$pa_G(i,j)=pa_G(i)\cup pa_G(j)\setminus\{i,j\},$
$ant(i,j)=ant(i)\cup ant(j)\setminus\{i,j\},$ and 
$pst(i,j)=pst(i)\cup pst(j)\setminus\{i,j\}.$
The distribution $\mathcal{P}$ of $(X_n)_{n\in V}$ satisfies a pairwise Markov property (Pm),  for $m = 1, 2, 3, 4$, with respect to MVR CG($G$) if for every uncoupled pair of nodes
$i$ and $j$ (i.e., there is no directed or bidirected edge between $i$ and $j$): 

	\noindent (P1):\label{Pair1} $i\!\perp\!\!\!\perp j|pst(i,j)\quad$, (P2):\label{Pair2} $i\!\perp\!\!\!\perp j|ant(i,j)\quad$, (P3):\label{Pair3} $i\!\perp\!\!\!\perp j|pa_G(i,j)\quad$, and  (P4):\label{Pair4} $i\!\perp\!\!\!\perp j|pa_G(i)$ if $i\prec j$.
    
Notice that in (P4), $pa_G(i)$ may be replaced by $pa_G(j)$ whenever the two nodes are in the same connected component. Sadeghi and Wermuth in \citep{sw} proved that all of above mentioned pairwise Markov properties are equivalent for compositional graphoids. Also, they show that each one of
the above listed pairwise Markov properties is equivalent to the global Markov properties in Definitions \ref{RS}, \ref{RS2} \citep[Corollary 1]{sw}. The necessity of intersection and composition properties follows from \citep[Section 6.3]{sl}.

\subsection{Block-recursive, Multivariate Regression (MR), and Ordered Local Markov Properties}
\begin{definition}
	Given a chain graph $G$, the set $Nb_G (A)$ is the union of $A$ itself and the set of nodes $w$ that are neighbors of $A$, that is, coupled by a bidirected edge to some node $v$ in $A$. Moreover, the set of non-descendants $nd_D(T)$ of a chain component $T$, is the union of all components $T'$ such that there is no directed path from $T$ to $T'$ in the directed graph of chain components $D$.
\end{definition} 
\begin{definition}\label{def1}\emph{(multivariate regression (MR) Markov property for MVR CGs \citep{ml2})}\footnote{A generalization of this property for regression graphs is the ordered regression graph Markov property in \citep{rov}.}
	Let G be a chain graph with chain components $(T|T\in\mathcal{T})$. A joint distribution P of the random vector X obeys the \emph{multivariate regression (MR) Markov property} with respect to $G$ if it satisfies the following independences. For all $T\in \mathcal{T}$ and for all $A \subseteq T$:
	
		\noindent (MR1)  if A is connected:$A\!\perp\!\!\!\perp [pre(T)\setminus  pa_G(A)]|pa_G(A)$.
		
		\noindent (MR2) if $A$ is disconnected with connected components $A_1,\dots,A_r$: $A_1\!\perp\!\!\!\perp \dots \!\perp\!\!\!\perp A_r|pre(T)$. 
	
\end{definition}
\begin{remark}{\label{rem1}}\citep[Remark 2]{ml2}
	One immediate consequence of Definition \ref{def1} is that if the probability density p(x) is strictly positive, then it factorizes according to the directed acyclic graph of the chain components: $p(x)=\prod_{T\in\mathcal{T}}p(x_T|x_{pa_D(T)}).$
\end{remark}
\begin{definition}\label{BR}\emph{(Chain graph Markov property of type IV \citep{d})}
	Let G be a chain graph with chain components $(T | T \in\mathcal{T})$ and directed acyclic graph $D$ of components. The joint probability distribution of $X$ obeys \emph{the block-recursive Markov property of type IV} if it satisfies the following independencies:
	
		\noindent (IV0): $T\!\perp\!\!\!\perp [nd_D(T)\setminus  pa_D(T)]|pa_D(T)$, for all $T\in \mathcal{T}$;
		
		\noindent (IV1): $A\!\perp\!\!\!\perp [pa_D(T)\setminus  pa_G(A)]|pa_G(A)$, for all $T\in \mathcal{T}$, and for all $A\subseteq T$;
		
		\noindent (IV2): $A\!\perp\!\!\!\perp [T\setminus Nb_G(A)]|pa_D(T)$, for all $T\in \mathcal{T}$, and for all connected subsets $A\subseteq T.$
	
\end{definition}
The following example shows that independence models, in general, resulting from Definitions \ref{def1}, \ref{BR} are different.
\begin{example}
	Consider the MVR chain graph $G$ in Figure \ref{Fig:11}.
	\begin{figure}[ht]
		\centering
		\includegraphics[scale=.75]{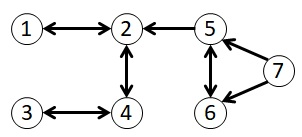}
		\caption{an MVR CG with chain components: $\mathcal{T}=\{T_1=\{1,2,3,4\} , T_2=\{5,6\}, T_3=\{7\}\}$.} \label{Fig:11}
	\end{figure}
For the connected set $A=\{1,2\}$ the condition (MR1) implies that $1,2\!\perp\!\!\!\perp 6,7|5$ while the condition (IV2) implies that $1,2\!\perp\!\!\!\perp 6|5$, which is not implied directly by (MR1) and (MR2). Also, the condition (MR2) implies that $1\!\perp\!\!\!\perp 3,4|5,6,7$ while the condition (IV2) implies that $1\!\perp\!\!\!\perp 3,4|5,6$, which is not implied directly by (MR1) and (MR2).
\end{example}
Theorem 1 in \citep{ml2} states that for a given chain graph $G$, the multivariate regression Markov property is equivalent to the block-recursive Markov property of type IV. Also, Drton in \citep[Section 7 Discussion]{d} claims (without proof) that the block-recursive Markov property of type IV can be
shown to be equivalent to the global Markov property proposed in \citep{rs, r2}. 

Now, we introduce a local Markov property for ADMGs proposed by Richardson in \citep{r2}, which is an extension of the
local well-numbering Markov property for DAGs introduced in \citep{ldll}. For this purpose, we need to consider the following definitions and notations:
\begin{definition}
	For a given acyclic directed mixed graph (ADMG) G, the induced bidirected graph
	$(G)_\leftrightarrow$ is the graph formed by removing all directed edges from G.
	The district (aka c-component) for a vertex x in G is the connected component of x in $(G)_\leftrightarrow$, or equivalently $$dis_G(x)=\{y | y\leftrightarrow\dots\leftrightarrow x \textrm{ in G, or }x=y\}.$$ As usual we apply the definition disjunctively to sets: $dis_A(B)=\cup_{x\in B}dis_A(x).$ A set C is path-connected in $(G)_\leftrightarrow$ if every pair of vertices in C are connected via a path in $(G)_\leftrightarrow$; equivalently, every vertex in C has the same district in G.
\end{definition}
\begin{definition}
	In an ADMG, a set A is said to be ancestrally closed if
	$x\rightarrow\dots\rightarrow a$ in G with $a\in A$ implies that $x\in A$. The set of ancestrally closed sets is defined as follows:  $$\mathcal{A}(G)=\{A|an_G(A)=A\}.$$
	If $A$ is an ancestrally closed set in an ADMG ($G$), and $x$ is a vertex in $A$ that has no children in $A$ then we define the \emph{Markov blanket} of a vertex $x$ with respect to the induced subgraph on $A$ as
	$$mb(x,A)= pa_{G}(dis_{G_A}(x))\cup (dis_{G_A}(x)\setminus\{x\}),$$
	where $dis_{G_A}$ is the district of $x$ in the induced subgraph $G_A$.
\end{definition} 
\begin{definition}
	Let $G$ be an acyclic directed mixed graph. Specify a total ordering ($\prec$) on the vertices
	of $G$, such that $x\prec y \Rightarrow y\not\in an(x)$; such an ordering is said to be consistent with $G$. Define $pre_{G,\prec}(x)=\{v|v\prec x \textrm{ or } v=x\}.$
\end{definition}
\begin{definition}[Ordered local Markov property]\label{lo}
	Let $G$ be an acyclic directed mixed graph. An independence model $\Im$ over the node
	set of $G$ satisfies the ordered local Markov property for $G$, with
	respect to the ordering $\prec$, if for any $x$, and ancestrally closed set $A$ such that $x\in A\subseteq pre_{G,\prec}(x)$,
	$$\{x\}\perp\!\!\!\perp [A\setminus(mb(x,A)\cup \{x\})]|mb(x,A).$$ 
\end{definition}
Since MVR chain graphs are a subclass of ADMGs, the ordered local Markov property in Definition \ref{lo} can be used as a local Markov property for MVR chain graphs.

Five of the Markov properties introduced in this and the previous subsection are equivalent for all probability distributions, as shown in the following theorem.

\begin{theorem}\label{thm1}
	Let $G$ be an MVR chain graph. For an independence model $\Im$ over the node
	set of $G$, the following conditions are equivalent:
	
	\noindent(i) $\Im$ satisfies the global Markov property w.r.t. $G$ in Definition \ref{RS};
	
	\noindent(ii) $\Im$ satisfies the global Markov property w.r.t. $G$ in Definition \ref{RS2};
	
	\noindent(iii) $\Im$ satisfies the block recursive Markov property w.r.t. $G$ in Definition \ref{BR};
	
	\noindent(iv) $\Im$ satisfies the MR Markov property w.r.t. $G$ in Definition \ref{def1}.
	
	\noindent(v) $\Im$ satisfies the ordered local Markov property w.r.t. $G$ in Definition \ref{lo}.
\end{theorem}
\begin{proof}
	See Appendix A for the proof of this theorem.
\end{proof}

\subsection{An Alternative Local Markov Property for MVR Chain Graphs}
In this subsection we
formulate an alternative local Markov property for MVR chain graphs. This property is different from and much more concise than the ordered Markov property proposed in \citep{r2}. The new local Markov property can be used to parameterize distributions efficiently when MVR chain graphs are learned from data, as done, for example, in \citep[Lemma 9]{jv3}.  While the new local Markov property is not equivalent to the five ones in Theorem \ref{thm1} in general, we show that it is equivalent to the global and ordered local Markov properties of MVR chain graphs for compositional graphoids.
\begin{definition}
	If there is a bidirected edge between vertices $u$ and $v$, $u$ and $v$ are said to be neighbors. The boundary $bd(u)$ of a vertex $u$ is the set of vertices in $V\setminus \{u\}$ that are parents or neighbors of vertex $u$. The descendants of vertex $u$ are $de(u)=\{v| u \textrm{ is an ancestor of } v\}$. The non-descendants of vertex $u$ are $nd(u)= V\setminus (de(u)\cup \{u\})$.
\end{definition}

\begin{definition}\label{local}
	The local Markov property for an MVR chain graph $G$ with vertex set $V$ holds if, for every $v\in V$: 
	$v\perp\!\!\!\perp [nd(v)\setminus bd(v)]|pa_G(v).$
\end{definition}
\begin{remark}
	In DAGs, $bd(v)=pa_G(v)$, and the local Markov property given above reduces to the directed local Markov property introduced by  Lauritzen et al. in \citep{ldll}. Also, in  covariance graphs \footnote{Equivalently, bidirected graphs, as explained in \citep[section 4.1]{r2}.} the local Markov property given above reduces to the dual local Markov property introduced by Kauermann in \citep[Definition 2.1]{k}.
\end{remark}

\begin{theorem}\label{thmlo}
	Let $G$ be an MVR chain graph. If an independence model $\Im$ over the node
	set of G is a compositional semi-graphoid, then $\Im$ satisfies the alternative local Markov property w.r.t. $G$ in Definition \ref{local} if
	and only if it satisfies the global Markov property w.r.t. $G$ in Definition \ref{RS2}.
\end{theorem}
\begin{proof}
	\noindent $(Global\Rightarrow Local)$: Let $X=\{v\},    Y=nd(v)\setminus bd(v), \textrm{ and } Z=pa_G(v)$. So, $an(X\cup Y\cup S)=v \cup (nd(v)\setminus bd(v))\cup pa_G(v)$ is an ancestor set, and $pa_G(v)$ separates $v$ from $nd(v)\setminus bd(v)$ in $(G_{v \cup (nd(v)\setminus bd(v))\cup pa_G(v)})^a$; this shows that the global Markov property in Definition \ref{RS2} implies the local Markov property in Definition \ref{local}.
	
	\noindent $(Local\Rightarrow MR)$: We prove this by considering the following two cases:
	
	\noindent Case 1): Let $A\subseteq T$ is connected. Using the alternative local Markov property for each $x\in A$ implies that: $\{x\} \perp\!\!\!\perp [nd(x)\setminus bd(x)]|pa_G(x)$. Since $(pre(T)\setminus pa_G(A))\subseteq (nd(x)\setminus bd(x))$, using the decomposition and weak union property give: $\{x\} \perp\!\!\!\perp (pre(T)\setminus pa_G(A))|pa_G(A)\textrm{, for all } x\in A$. Using the composition property leads to (MR1): $A \perp\!\!\!\perp (pre(T)\setminus pa_G(A))|pa_G(A)$.
	
	\noindent Case 2): Let $A\subseteq T$ is disconnected with connected components $A_1,\dots,A_r$. For $1\le i\ne j\le r$ we have: $\{x\} \perp\!\!\!\perp [nd(x)\setminus bd(x)]|pa_G(x) \textrm{, for all } x\in A_i$. Since $[(pre(T)\setminus pa_G(A))\cup A_j]\subseteq (nd(x)\setminus bd(x))$, using the decomposition and weak union property give: $\{x\} \perp\!\!\!\perp A_j| pre(T)\textrm{, for all } x\in A_i$. Using the composition property leads to (MR2): $A_i\perp\!\!\!\perp  A_j|pre(T), \textrm{ for all } 1\le i\ne j\le r$.
	
	\noindent $(MR\Rightarrow Global)$: The result follows from Theorem \ref{thm1}.
\end{proof}
The necessity of composition property in Theorem \ref{thmlo} follows from the fact that local and global Markov properties for bidirected graphs, which are a subclass of MVR CGs, are equivalent only for compositional semi-graphoids \citep{k,br}.

\section{An Alternative Factorization for MVR Chain Graphs}\label{sec5}
According to the definition of MVR chain graphs, it is obvious that they are a subclass of acyclic directed mixed graphs (ADMGs). In this section, we derive an explicit factorization criterion for MVR chain graphs based on the proposed factorization criterion for acyclic directed mixed graphs in \citep{er}. For this purpose, we need to consider the following definition and notations:

\begin{definition}\label{def5}
	An ordered
	pair of sets $(H,T)$ form the head and tail of a term
	associated with an ADMG $G$ if and only if all of the following
	hold:
	
	\noindent 1. $H=barren(H)$, where $barren(H)=\{v\in H|de(v)\cap H=\{v\}\}$.
	
	\noindent 2. H contained within a single district of $G_{an(H)}$. 
	
	\noindent 3. $T=tail(H)=(dis_{an(H)}(H)\setminus H)\cup pa(dis_{an(H)}(H)).$
\end{definition}
Evans and Richardson in \citep[Theorem 4.12]{er} prove that a probability distribution $P$ obeys the
global Markov property for an ADMG($G$) if and only if for every
$A\in\mathcal{A}(G)$,
\begin{equation}\label{eq_r}
p(X_A)=\prod_{H\in [A]_G}p(X_H|tail(H)),
\end{equation}  
where $[A]_G$ denotes a partition of A into sets $\{H_1,\dots,H_k\}\subseteq \mathcal{H}(G)$ (for a graph $G$, the set of heads is denoted by $\mathcal{H}(G)$), defined with $tail(H)$, as above. The following theorem provides an alternative factorization criterion for MVR chain graphs based on the proposed factorization criterion for acyclic directed mixed graphs in \citep{er}. 

\begin{theorem}\label{thm5}
	Let G be an MVR chain graph with chain components $(T|T\in\mathcal{T})$. If a probability distribution P obeys the
	global Markov property for G then 
	$p(x)=\prod_{T\in\mathcal{T}}p(x_T|x_{pa_G(T)}).$  
\end{theorem}
\begin{proof}
	According to Theorem 4.12 in \citep{er}, since $G\in \mathcal{A}(G)$, it is enough to show that $\mathcal{H}(G)=\{T|T\in\mathcal{T}\}$ and $tail(T)=pa_G(T)$, where $T\in\mathcal{T}$. In other words, it is enough to show that for every $T$ in $\mathcal{T}$, $(T, pa_G(T))$ satisfies the three conditions in Definition \ref{def5}. 
	
	\noindent 1. Let $x, y \in T$ and $T\in \mathcal{T}$. Then $y$ is not a descendant of $x$. Also, we know that $x\in de(x)$, by definition. Therefore, $T=barren(T).$
	
	\noindent 2. Let $T\in\mathcal{T}$, then from the definitions of an MVR chain graph and induced bidirected graph, it is obvious that $T$ is a single connected component of the forest $(G_{an(T)})_\leftrightarrow$. So, $T$ contained within a single district of $(G_{an(T)})_\leftrightarrow$.
	
	\noindent 3. $T\subseteq an(T)$ by definition. So, $\forall x\in T: dis_{an(T)}(x)=\{y|y\leftrightarrow\dots\leftrightarrow x \textrm{ in }an(T), \textrm{or }x=y\}=T$. Therefore, $dis_{an(T)}(T)=T$ and $dis_{an(T)}(T)\setminus T=\emptyset$. In other words, $tail(T)=pa_G(T)$.
\end{proof}
\begin{example}
	Consider the MVR chain graph G in Example \ref{Fig:11}. 
	Since $[G]_G=\{\{1,2,3,4\}\{5,6\}\{7\}\}$ so, $tail(\{1,2,3,4\})=\{5\}, tail(\{5,6\})=\{7\},$ and $tail(\{7\})=\emptyset$. Therefore, based on Theorem \ref{thm5} we have: $p=p_{1234|5}  p_{56|7}  p_{7}$. However, the corresponding factorization of G based on the formula in \citep{d, ml2} is: $p=p_{1234|56}  p_{56|7}  p_{7}$. 
\end{example}

The advantage of the new factorization is that it requires only graphical parents, rather than parent components in each factor, resulting in smaller variable sets for each factor, and therefore speeding up belief propagation.  Moreover, the new factorization is the same as the outer factorization of LWF and AMP CGs, as described in \citep{l,lr,cdls,amp}.

\section{Intervention in MVR Chain Graphs}
In the absence of a theory of intervention for chain graphs, a researcher would be unable to answer questions concerning the consequences 
of intervening in a system with the structure of a chain graph \citep{r1}. Fortunately, an intuitive account of the causal interpretation of MVR chain graphs is as follows. 
We interpret the edge $A\to B$ as $A$ being a cause of $B$. We interpret the edge $A\leftrightarrow B$ as $A$ and $B$
having an unobserved common cause $\lambda_{AB}$, i.e. a confounder. 

Given the above causal interpretation of an MVR CG $G$, intervening on $X\subseteq V$ so that $X$ is
no longer under the influence of its usual causes amounts to replacing the right-hand side of the
equations for the random variables in $X$ with expressions that do not involve their usual causes and normalizing. 
Graphically, it amounts to modifying $G$ as follows. Delete from $G$ all the edges $A\to B$ and $A\leftrightarrow B$ with $B\in X$ \citep{pena}.
\section*{Conclusion and Summary}
Based on the interpretation of the type of edges in a chain graph, there are different conditional independence structures among random variables in the corresponding probabilistic model. Other than pairwise Markov properties, we showed that for MVR chain graphs all Markov properties in the literature are equivalent for semi-graphoids. We proposed an alternative local Markov property for MVR chain graphs, and we proved that it is equivalent to other Markov properties for compositional semi-graphoids. Also, we obtained an alternative formula for factorization of an MVR chain graph. Table \ref{t3} summarizes some of the most important attributes of different types of common interpretations of chain graphs.



\begin{table}[bt]
	\centering
    \begin{threeparttable}
    \footnotesize
	\begin{tabular}{| p{25mm} | p{23mm} | p{25mm} | p{35mm} | p{23mm} |} 
		\hline
		\centering{Type of chain graph} & Does it represent   independence model of DAGs under \newline marginalization? & \centering{Global Markov property} & \centering{Factorization of $p(x)$} & \centering{Model selection (structural learning) algorithm(s) [constraint based method]} \tabularnewline
        \hline
		MVR CGs:  Cox \& Wermuth \citep{cw1,cw2,wc},  Pe{\~n}a \& Sonntag \citep{p2, s}, Sadeghi \& Lauritzen \citep{sl}, Drton (type IV) \citep{d}, Marchetti \& Lupparelli \citep{ml1,ml2} & Yes (claimed in \citep{cw2,ws,sl, s}, proved in Corollary \ref{cor1}) & 
		$(1) \quad X\!\perp\!\!\!\perp Y|Z$ if $X$ is separated from $Y$ by $Z$ in $(G_{ant(X\cup Y\cup Z)})^a$ or $(G_{an(X\cup Y\cup Z)})^a$ \citep{r2,rs}.
		\vskip 0.1in
		(2) $X\!\perp\!\!\!\perp Y|Z$ if $X$ is separated from $Y$ by $Z$ in $(G_{Antec(X\cup Y\cup Z)})^a.$
		\vskip 0.1in
		(1) and (2) are equivalent for compositional graphoids (see supplementary material).
		& (1) {Theorem\quad\ref{thm5},} $$\prod_{T\in\mathcal{T}}p(x_T|x_{pa(T)})$$  
		
		\vskip 0.1in

		(2) $$\prod_{T\in\mathcal{T}}p(x_T|x_{pa_D(T)})$$
		where $pa_D(T)$ is the union of all the chain components that are parents of $T$ in the directed graph $D$ {\citep{d,ml2}}.
		& PC like algorithm for MVR CGs in \citep{s,sp}, Decomposition-based algorithm for MVR CGs in \citep{jv3}. \\
		\hline
		
		LWF CGs \citep{f,lw}, Drton (type I) \citep{d} & \centering{No} & $X\!\perp\!\!\!\perp Y|Z$ if $X$ is separated from $Y$ by $Z$ in $(G_{An(X\cup Y\cup Z)})^m$ \citep{l}. & \citep{cdls,lr} $$\prod_{\tau\in\mathcal{T}}p(x_{\tau}|x_{pa(\tau)}),$$ where $p(x_{\tau}|x_{pa(\tau)})= Z^{-1}(x_{pa(\tau)})\prod_{c\in C}\phi_c(x_c),$ where $C$ are the complete sets in the moral graph $(\tau\cup pa(\tau))^m.$  & IC like algorithm in \citep{srs}, LCD algorithm in \citep{mxg}, CKES algorithm in \citep{psn,s} \\
		\hline
		
		AMP CGs \citep{amp}, Drton (type II) \citep{d} & \centering{No} & $X\!\perp\!\!\!\perp Y|Z$ if $X$ is separated from $Y$ by $Z$ in the undirected 
		graph $Aug[CG;X,Y,Z]$ \citep{r1}. & $\prod_{\tau\in\mathcal{T}}p(x_{\tau}|x_{pa(\tau)}),$ where no further factorization similar to LWF model appears to hold in general \citep{amp}. For the positive distribution $p$ see \citep{p3}. & PC like algorithm in \citep{p1} \\  
		\hline
	\end{tabular}\caption{Properties of chain graphs under different interpretations}\label{t3}
    \end{threeparttable}
\end{table}

\section*{acknowledgements}
This work has been partially supported by Office of Naval Research grant ONR N00014-17-1-2842. This research is based upon work supported in part by the Office of the Director of National
Intelligence (ODNI), Intelligence Advanced Research Projects Activity (IARPA), award/contract
number 2017-16112300009. The views and conclusions contained therein are those of the authors
and should not be interpreted as necessarily representing the official policies, either expressed or implied,
of ODNI, IARPA, or the U.S. Government. The U.S. Government is authorized to reproduce
and distribute reprints for governmental purposes, notwithstanding annotation therein.

An early version of this work was presented at the workshop of the Ninth International Conference on Probabilistic Graphical Models, Prague, September 11-14, 2018. Comments by reviewers and workshop participants are gratefully acknowledged.

\section*{Appendix A. Proof of Theorem \ref{thm1}}
\begin{proof}
	\noindent \textbf{(i)$\Rightarrow$(ii):} This has already been proved in \citep[Theorem 1]{r2}.
	
	\noindent \textbf{(ii)$\Rightarrow$(iii):}  Assume that the independence model $\Im$ over the node
	set of MVR CG($G$) satisfies the global Markov property w.r.t. $G$ in Definition \ref{RS2}. We have the following three cases:
	
	\noindent Case 1: Let $X=\tau\in \mathcal{T},    Y=nd_D(\tau)\setminus pa_D(\tau), \textrm{ and } Z=pa_D(\tau)$. So, $an(X\cup Y\cup Z)=\tau \cup nd_D(\tau)$ is an ancestor set, and $pa_D(\tau)$ separates $\tau$ from $nd_D(\tau)\setminus pa_D(\tau)$ in $(G_{\tau \cup nd_D(\tau)})^a$; this shows that the global Markov property in Definition \ref{RS2} implies (IV0) in Definition \ref{BR}.
	
	\noindent Case 2: Assume that $X=\sigma \subseteq \tau\in \mathcal{T}, Y=pa_D(\tau)\setminus pa_G(\sigma), \textrm{ and } Z=pa_G(\sigma)$. Consider that $W=an(X\cup Y\cup Z)=an(\sigma\cup pa_D(\tau))$. We know that there is no directed edge from $pa_D(\tau)\setminus pa_G(\sigma)$ to elements of $\sigma$, and also there is no collider path between nodes of  $Y$ and  $\sigma$ in $W$. So, every connecting path that connects $pa_D(\tau)\setminus pa_G(\sigma)$ to $\sigma$ in $(G_W)^a$ has intersection with $pa_G(\sigma)$, which means $pa_G(\sigma)$ separates $pa_D(\tau)\setminus pa_G(\sigma)$ from $\sigma$ in $(G_W)^a$; this shows that the global Markov property in Definition \ref{RS2} implies (IV1) in Definition \ref{BR}.
	
	\noindent Case 3: Assume that $X=\sigma \subsetneq \tau\in \mathcal{T}$ is a connected subset of $\tau$. Also, assume that $Y=\tau\setminus Nb_G(\sigma), \textrm{ and } Z=pa_D(\tau)$. Obviously, $\sigma$ and $\tau\setminus Nb_G(\sigma)$ are two subsets of $\tau$ such that there is no connection between their elements.  Consider that $A$ is the ancestor set containing $\sigma$, $\tau\setminus Nb_G(\sigma)$, and $pa_D(\tau)$. Clearly, $pa_D(\tau)\subseteq A$. Since $\sigma$ and $\tau\setminus Nb_G(\sigma)$ are disconnected in $\tau$, so any connecting path between them in $A$ (if it exists) must pass through $pa_D(\tau)$ in $(G_A)^a$; this shows that the global Markov property in Definition \ref{RS2} implies (IV2) in Definition \ref{BR}.
	
	\noindent \textbf{(iii)$\Rightarrow$(iv):} Assume that the independence model $\Im$ over the node
	set of MVR CG($G$) satisfies the block recursive Markov property w.r.t. $G$ in Definition \ref{BR}. We show that $\Im$ satisfies the MR Markov property w.r.t. $G$ in Definition \ref{def1} by considering the following two cases:
	
	\noindent Case 1 (IV0 and IV1 $\Rightarrow$ MR1): Assume that $A$ is a connected subset of $\tau$. From (IV1) we have:
	\begin{equation}\label{eq3}
	A\perp\!\!\!\perp (pa_D(\tau)\setminus pa_G(A))|pa_G(A)
	\end{equation}
	Also, from (IV0) we have $\tau \perp\!\!\!\perp (nd_D(\tau)\setminus pa_D(\tau))|pa_D(\tau)$, the decomposition property implies that 
	\begin{equation}\label{eq4}
	A \perp\!\!\!\perp (nd_D(\tau)\setminus pa_D(\tau))|pa_D(\tau)
	\end{equation}
	Using the contraction property for (\ref{eq3}) and (\ref{eq4}) gives: $A \perp\!\!\!\perp [(nd_D(\tau)\setminus pa_G(\tau))\cup (pa_D(\tau\setminus pa_G(A)))]|pa_G(\tau).$ Using the decomposition property for this independence relationship gives (MR1): $A \perp\!\!\!\perp (pre(\tau)\setminus pa_G(A))|pa_G(A),$ because $(pre(\tau)\setminus pa_G(A))\subseteq [(nd_D(\tau)\setminus pa_G(\tau))\cup (pa_D(\tau\setminus pa_G(A)))]$.
	
	\noindent Case 2 (IV0 and IV2 $\Rightarrow$ MR2): Consider that $A$ is a disconnected subset of $\tau$ that contains $r$ connected components $A_1,\dots ,A_r$ i.e., $A=A_1\cup\dots\cup A_r$. From (IV2) we have: $A_1\perp\!\!\!\perp \tau\setminus Nb_G(A_1)|pa_D(\tau)$. Using the decomposition property gives: 
	\begin{equation}\label{eq5}
	A_1\perp\!\!\!\perp A_2|pa_D(\tau)
	\end{equation}
	Also, using decomposition for (IV0) gives: $(A_1\cup A_2)\perp\!\!\!\perp (pre(\tau)\setminus pa_D(\tau))|pa_D(\tau)$. Applying the weak union property for this independence relation gives: $ A_1\perp\!\!\!\perp (pre(\tau)\setminus pa_D(\tau))|[A_2\cup pa_D(\tau)]$. Using the contraction property for this and (\ref{eq5}) gives: $ A_1\perp\!\!\!\perp [A_2\cup (pre(\tau)\setminus pa_D(\tau))]|pa_D(\tau)$. Using the weak union property leads to $A_1\perp\!\!\!\perp A_2|[(pa_D(\tau)\cup(pre(\tau)\setminus pa_D(\tau)))=pre(\tau)]$. Similarly, we can prove that for every $1\leq i\neq j\leq r$: $A_i\perp\!\!\!\perp A_j|pre(\tau)$.
	
	\noindent \textbf{(iv)$\Rightarrow$(v):} Assume that the independence model $\Im$ over the node
	set of MVR CG($G$) satisfies the MR Markov property w.r.t. $G$ in Definition \ref{def1}, and $\prec$ is an ordering that is consistent with $G$. Let $x\in A\subseteq pre_{G,\prec}(x)$, We show that $\Im$ satisfies the ordered local Markov property w.r.t. $G$ in Definition \ref{lo} by considering the following two cases:
	
	\noindent Case 1: There is a chain component $T$ such that $x\in T$. Consider that $A\cap T$ is a connected subset of $T$. From (MR1) we have: $dis_{G_A}(x)\perp\!\!\!\perp [pre(T)\setminus pa_G(dis_{G_A}(x))]|pa_G(dis_{G_A}(x))$. Using the weak union property gives: $\{x\} \perp\!\!\!\perp [pre(T)\setminus pa_G(dis_{G_A}(x))]|[pa_G(dis_{G_A}(x))\cup(dis_{G_A}(x)\setminus\{x\})]$. Since $[A\setminus (mb(x,A)\cup \{x\})] \subseteq [pre(T)\setminus pa_G(dis_{G_A}(x))]$, using the decomposition property leads to: $\{x\} \perp\!\!\!\perp [A\setminus (mb(x,A)\cup \{x\})]|mb(x,A)$. 
	
	\noindent Case 2: There is a chain component $T$ such that $x\in T$, and $A\cap T$ is a disconnected subset of $T$ with connected components $A_1,\dots,A_k$ i.e., $A\cap T = A_1\cup \dots\cup A_k$. It is clear that there is a $1\le d\le k$ such that $A_d=dis_{G_A}(x)$. We have the following two sub-cases:
	
	\noindent Sub-case I): $\sigma := T\setminus Nb_G(A_d)$ is a connected subset of $T$.
	\begin{equation}\label{eq0}
	\left\{
	\begin{array}{l}  
	\textrm{From (MR2): } A_d\perp\!\!\!\perp\sigma|pre(T)  \\
	\textrm{From (MR1): } A_d\perp\!\!\!\perp (pre(T)\setminus pa_G(A_d))|pa_G(A_d) 
	\end{array} \right.
	\end{equation}
	Using the contraction property for (\ref{eq0}) gives: $A_d\perp\!\!\!\perp [\sigma\cup(pre(T)\setminus pa_G(A_d))]|pa_G(A_d)$. Using the weak union property gives: $\{x\} \perp\!\!\!\perp [pre(T)\setminus pa_G(dis_{G_A}(x))]|[pa_G(dis_{G_A}(x))\cup(dis_{G_A}(x)\setminus\{x\})]$. Since $[A\setminus (mb(x,A)\cup \{x\})] \subseteq [pre(T)\setminus pa_G(dis_{G_A}(x))]$, using the decomposition property leads to: $\{x\} \perp\!\!\!\perp [A\setminus (mb(x,A)\cup \{x\})]|mb(x,A)$.
	
	\noindent Sub-case II): $T\setminus Nb_G(A_d)$ is a disconnected subset of $T$ with connected component $\sigma_1, \sigma_2$ i.e., $T\setminus Nb_G(A_d)=\sigma_1\cup \sigma_2$. From (MR1) we have: $\sigma_1\perp\!\!\!\perp (T\setminus Nb_G(\sigma_1))|pre(T)$. Since $(A_d\cup\sigma_2)\subseteq (T\setminus Nb_G(\sigma_1))$, using the decomposition and weak union property give: $\sigma_1\perp\!\!\!\perp  A_d|(pre(T)\cup\sigma_2)$. Using the symmetry property implies that $A_d\perp\!\!\!\perp \sigma_1|(pre(T)\cup\sigma_2)$.
	\begin{equation}\label{eq00}
	\left\{
	\begin{array}{l}  
	A_d\perp\!\!\!\perp \sigma_1|(pre(T)\cup\sigma_2)  \\
	\textrm{From (MR2): } A_d\perp\!\!\!\perp \sigma_2|pre(T) 
	\end{array} \right.
	\end{equation}
	Using the contraction property for (\ref{eq00}) gives: $A_d\perp\!\!\!\perp (\sigma_1\cup\sigma_2)|pre(T)$.
	\begin{equation}\label{eq000}
	\left\{
	\begin{array}{l}  
	A_d\perp\!\!\!\perp (\sigma_1\cup\sigma_2)|pre(T)  \\
	\textrm{From (MR1): } A_d\perp\!\!\!\perp (pre(T)\setminus pa_G(A_d)|pa_G(A_d) 
	\end{array} \right.
	\end{equation}
	Using the contraction property for (\ref{eq000}) gives: $A_d\perp\!\!\!\perp [(\sigma_1\cup\sigma_2)\cup(pre(T)\setminus pa_G(A_d))]|pa_G(A_d)$. Using the decomposition property gives: $\{x\}\perp\!\!\!\perp [(\sigma_1\cup\sigma_2)\cup(pre(T)\setminus pa_G(A_d))]|mb(x,A)$. Since $[A\setminus (mb(x,A)\cup \{x\})] \subseteq [(\sigma_1\cup\sigma_2)\cup(pre(T)\setminus pa_G(A_d))]$, using the decomposition property leads to: $\{x\} \perp\!\!\!\perp [A\setminus (mb(x,A)\cup \{x\})]|mb(x,A)$. 
	
	\noindent \textbf{(v)$\Rightarrow$(i):}	This has already been proved in \citep[Theorem 2]{r2}.	
\end{proof}



\bibliography{reference}



\end{document}